\documentclass[12pt]{article}

\usepackage[top=1in, bottom=1.25in, left=1in, right=1in]{geometry}

\usepackage{balance}

\usepackage[pdftex,
	bookmarks=false,
	colorlinks=true,
	urlcolor=red,
	linkcolor=blue,
	citecolor=red,
	pdfstartpage=1,
	pdfstartview={FitH},  
	bookmarksopen=false
	]{hyperref}

\usepackage[english]{babel}
\usepackage[toc,page]{appendix}
\usepackage{hyperref}
\usepackage{wrapfig,lipsum}	
\usepackage[labelformat=empty]{caption}
\usepackage{yfonts}	
\usepackage{amssymb}
\usepackage{amsmath}
\usepackage{amsthm}
\usepackage{scalerel}
\usepackage{verbatim}
\usepackage{thm-restate}

\newtheorem{theorem}{Theorem}[section]
\newtheorem{lemma}[theorem]{Lemma}

\newtheorem{proposition}[theorem]{Proposition}
\newtheorem{corollary}[theorem]{Corollary}

\theoremstyle{definition}
\newtheorem{definition}{Definition}
\newtheorem{remark}[theorem]{Remark}

\numberwithin{equation}{section}

\usepackage{times}

\newcommand{\bea}{\begin{eqnarray}}
\newcommand{\eea}{\end{eqnarray}}
\newcommand{\<}{\langle}
\renewcommand{\>}{\rangle}

\newcommand{\wt}{\widetilde}

\def\Vol{{\rm Vol}}

\def\poly{{\rm poly}}

\def\eps{{\varepsilon}}

\def\supp{{\rm supp}}

\def\cF{{\mathcal F}}
\def\cG{{\mathcal G}}

\def\<{\langle}
\def\>{\rangle}

\def\cN{{\cal N}}

\def\cU{{\cal U}}

\def\cE{{\mathcal E}}

\def\b0{{\boldsymbol{0}}}

\def\Ber{{\sf Ber}}
\def\Var{{\rm Var}}

\def\cA{{\mathcal A}}

\def\cK{{\mathcal K}}

\def\olambda{\overline{\lambda}}
\def\ulambda{\underline{\lambda}}

\renewcommand{\b}{\mathbf{b}}

\def\lt{\left}
\def\rt{\right}

\def\eps{\varepsilon}

\def\bbA{{\mathbb{A}}}

\def\bbE{{\mathbb{E}}}

\def\bbP{{\mathbb{P}}}
\def\bbR{{\mathbb{R}}}
\def\bbS{{\mathbb{S}}}

\def\cF{{\mathcal{F}}}
\def\cK{{\mathcal{K}}}
\def\cN{{\mathcal{N}}}
\def\cP{{\mathcal{P}}}

\def\width{{\mathrm{width}}}

\usepackage{slivkins-setup}
\usepackage{nicefrac}

\newcommand{\term}[1]{\text{\tt{#1}}\xspace}

\usepackage[lined,boxed,ruled,norelsize,algo2e,linesnumbered,noend]{algorithm2e}

\newcommand{\MainALG}{\term{Semibandit BIC Exploration}}

\newcommand{\ZEROS}{\term{ZEROS}}
\newcommand{\ExplorePhase}[1][j]{\textsc{Exploration Phase} for arm $#1$}
\newcommand{\ExploitPhase}[1][N]{\textsc{Exploitation Phase} with depth $#1$\xspace}
\newcommand{\PaddedPhase}{\textsc{Padded Phase}\xspace}

\newcommand{\padding}{\lambda}
\newcommand{\padG}{G_{\term{pad}}}

\usepackage{float}

\usepackage[suppress]{color-edits}  
\addauthor{rev}{red}   
\addauthor{ms}{blue}

\title{Incentivizing Exploration with Linear Contexts \\ and Combinatorial Actions}

\author{Mark Sellke
}

\date{}

\begin{document}

\maketitle

\begin{abstract}
\noindent We advance the study of incentivized bandit exploration, in which arm choices are viewed as recommendations and are required to be Bayesian incentive compatible. Recently \cite{sellke-slivkins} showed under certain independence assumptions that after collecting enough initial samples, the popular Thompson sampling algorithm becomes incentive compatible. We give an analog of this result for linear bandits, where the independence of the prior is replaced by a natural convexity condition. This opens up the possibility of efficient and regret-optimal incentivized exploration in high-dimensional action spaces. In the semibandit model, we also improve the sample complexity for the pre-Thompson sampling phase of initial data collection.
\end{abstract}

\section{Introduction}

\subsection{Problem Formulation and Results}

We consider \textbf{incentivized} exploration as introduced in \cite{kremer2014implementing,ICexploration-ec15}. These are bandit problems motivated by the scenario that the $t$-th action is a \textbf{recommendation} made to the $t$-th customer. A planner, representing a central service such as a major website, observes outcomes of past customers. By contrast, each customer sees only the recommendation they are given and visits the planner only once. Assuming a common Bayesian prior (and indistinguishability of customers), we would like our recommendations to be trustworthy, in the sense that a rational customer will follow our recommendations. We have the following definition:

\begin{definition}
\label{def:BIC}
Let $\cA$ be an action set and $\mu$ a prior distribution over the mean reward function $\ell^*:\cA\to \bbR$. A bandit algorithm which recommends a sequence $(A^{(1)},A^{(2)},\dots,A^{(T)})$ of actions in $\cA$ is said to be \textbf{Bayesian incentive compatible} (BIC) if for each $t\in [T]$:
\begin{equation}
\label{eq:BIC-def}
\begin{aligned}
    &\bbE[\mu(A)-\mu(A')~|~A^{(t)}=A]\geq 0,
    \\
    &\quad\forall A,A'\in \cA\text{ such that}~\bbP[A^{(t)}=A]>0.
\end{aligned}
\end{equation}
We will often refer to the algorithm as the \emph{planner}, and $A^{(t)}$ as a \emph{recommendation} made to an \emph{agent}. Here the agent knows the (source code for the) planner's algorithm, and the value of $t$ (i.e. his place in line) but nothing about the online feedback. (In particular, each agent appears just once and never ``comes back''.) Thus, \eqref{eq:BIC-def} exactly states that rational agents will follow the planner's recommendation.
\end{definition}

This model was formulated in \cite{kremer2014implementing}, and is a multi-stage generalization of the well-studied Bayesian persuasion problem in information design \cite{bergemann2019information,kamenica2019bayesian}. Fundamental results on feasible explorability and vanishing regret for BIC algorithms were obtained in \cite{ICexploration-ec15,ICexplorationGames-ec16}.
In all cases, the fundamental principle is to use \emph{information asymmetry} to guide agent decisions toward exploration.

Let us remark that averaging over the choice of $A^{(t)}$ shows that for any BIC algorthm $\cA$ and fixed time $t$, 
\[
    \bbE[\mu(A^{(t)})]\geq \sup_{A\in\cA} \bbE[\mu(A)].
\]
That is, the recommended time $t$ action $A^{(t)}$  is always better on average than exploiting according to the prior. Hence as a special case, BIC algorithms are guaranteed to benefit \textbf{all} users as compared to naive exploitation without learning. Moreover an agent's knowledge of the exact time $t$ only makes the BIC guarantee stronger.

As mentioned, fundamental results on BIC bandit algorithms were obtained in \cite{ICexploration-ec15,ICexplorationGames-ec16}. 
\msedit{For instance, the former work gave ``hidden exploration'' algorithms which exploit with such high frequency that performing \emph{any} bandit algorithm on the remaining time-steps is BIC (given suitable assumptions on the prior).}
Unfortunately, all algorithms in these works had to pay exponentially large multiplicative factors in their regret compared to non-BIC bandit algorithms. This ``price of incentives'' was studied quantitatively in our recent work \cite{sellke-slivkins} with Slivkins for the case of independent arms. We proved the classical Thompson sampling algorithm \cite{Thompson-1933} is BIC without modification once a mild number of initial samples per arm have been collected, leading to a natural two-phase algorithm. The first phase collects these initial samples in a BIC way, while the second simply performs Thompson sampling. Note that Thompson sampling obeys many state-of-the-art regret bounds, both Bayesian \cite{Russo-MathOR-14,bubeck2013prior,TSfirstorder} and frequentist \cite{Shipra-colt12,Kaufmann-alt12,Shipra-aistats13,TorTS-nips19,lattimore2021mirror}. Thus the result of \cite{sellke-slivkins} implies that the additional regret forced by the BIC requirement is essentially additive, and at most the number of rounds needed to collect the needed initial samples in a BIC way (a quantity studied separately therein).

The results of \cite{sellke-slivkins} are limited to the case of independent arms, and it is therefore of interest to expand the range of models under which Thompson sampling leads to provable incentive compatibility. This is the goal of the present paper. Our main focus is on the \textbf{linear} bandit in dimension $d$. Here the action set $\cA\subseteq\bbR^d$ may be infinite, and the reward function $\ell^*:\cA\to\bbR$ is always linear. This is a natural next step and allows for richer correlations between actions.
\msedit{From the practical viewpoint, BIC guarantees are relevant for recommending restaurants, movies, hotels, and doctors (see e.g. \cite{slivkins-MABbook}, Chapter 11); in all of these settings it is desirable to leverage contextual information to obtain sample complexity scaling with the ambient data dimension rather than the potentially huge number of total actions.}
We make no assumptions of independence but instead require natural geometric conditions, e.g. that the prior $\mu$ for $\ell^*$ is uniformly random on a convex body $\cK$ with bounded aspect ratio. This covers a fairly broad range of scenarios, including centered Gaussian $\mu$ by homogeneity. 
\msedit{Indeed we consider the connection with such conditions as a significant contribution of this work; in \cite{sellke-slivkins,hu2022incentivizing} the FKG inequality is used crucially throughout but requires independence properties that are unavailable with linear contexts.}
As a regret benchmark, recall that for the Bayesian regret of Thompson sampling is known to be $O(d\sqrt{T\log T})$ by \cite[Theorem 2]{dong2018information} (see also \cite{Russo-MathOR-14}).
Interestingly, the \emph{frequentist} regret of Thompson sampling for linear bandits is known to be larger by a factor of $\sqrt{d}$ \cite{Shipra-icml13,hamidi2020worst}.

Relative to this near-optimal guarantee, our first main result Theorem~\ref{thm:main-linear-bandit} shows that the price of incentives is again additive rather than multiplicative. Namely, Thompson sampling is again BIC after obtaining $\poly(d)$ initial samples which are well-spread in a natural spectral sense. Technically, we require action sets be to $\eps$-separated so that the recommendation of a given $A\in\cA$ has non-negligible probability; this is a mild condition since any action set can be discretized before running Thompson sampling. Further as shown in Theorem~\ref{thm:main-GLM}, the result extends to the \textbf{generalized} linear bandit when the link function's derivative is bounded above and below.

Next, we provide two counterexamples. The former shows that Thompson sampling may be BIC at time $1$ but not time $2$ -- this may be surprising as it is impossible for the multi-armed bandit with independent arms (see \cite[Lemma 4.9]{sellke-slivkins}). The latter gives a natural example in which initial data collection provably needs $e^{\Omega(d)}$ samples.
\msedit{In particular, the latter counterexample illustrates that further geometric conditions on the prior and/or action set are necessary for $\poly(d)$ sample complexity of BIC initial exploration. We leave a further study of this interesting problem for future work, but remark that exogenous payments could be used in practice to obtain the required $\poly(d)$ initial samples (see some of the references in Subsection~\ref{subsec:other}).}

Finally, we give new results for incentivized exploration of the combinatorial semibandit. Here actions consist of subsets $A\subseteq [d]$ of at most $d$ independent atoms which each give separate feedback as well as rewards. This is another natural testing ground for correlated rewards and was the recent focus of \cite{hu2022incentivizing}; they showed that Thompson sampling is still BIC with enough initial samples. \cite{hu2022incentivizing} gave initial exploration algorithms to bound the additive regret increase from being BIC, but with exponentially large sample complexity in typical cases. We improve this latter aspect of their work by extending the framework of \cite{sellke-slivkins}, linking the initial sample complexity to the minimax value of a two-player zero-sum game.

\msedit{
\subsection{Other Relevant Work}
\label{subsec:other}

Here we mention a few other related works in the broad area of incentivized exploration. 
\cite{frazier2014incentivizing} studies a similar problem, again in a Bayesian setting. However in their work the full history is made public so there is no information asymmetry. Instead, incentivization is achieved via exogenous payments.
\cite{wang2018multi,agrawal2020incentivising,wang2021incentivizing} study incentivized exploration in a non-Bayesian setting, where empirical averages are used instead of posterior means, and again incentives are realized through payments. The latter two works also focus on linear contexts and show $\wt O(\sqrt{T})$ total payment suffices for incentivization.
By contrast Theorem~\ref{thm:main-linear-bandit} implies that Thompson sampling is BIC after a constant (i.e. $T$-independent) amount of initial exploration. If payments can be used for incentivization, this implies in particular that a constant amount of total payment suffices for incentivized exploration. However as just mentioned, our setting and assumptions differ from the aforementioned works in multiple ways.

\cite{kannan2017fairness} studied the power of exogenous payments to incentivize related notions of fairness in which better actions ought to be played with higher probability. Their model includes information asymmetry in the opposite direction; agents observe the full history while the planner might not. 
\cite{immorlica2020incentivizing} proposed incentivization via \emph{selective data disclosure} where agents observe a carefully chosen subset of the history. However their setting is not precisely Bayesian as they assume agents perform naive empirical averaging over the chosen subset rather than taking the planner's algorithm into account. In fact it can be shown by a form of the revelation princple (see \cite{slivkins-MABbook}, Chapter 11) that in the perfectly Bayesian setting, general planner-to-agent signals have no more power than simple action recommendations.
Finally our work and the ones mentioned above assume agents are identical; a model with heterogenous agents was studied in \cite{immorlica2019bayesian}.
}

\section{Preliminaries}

This paper considers two bandit models with correlated rewards. In both cases the unknown optimal action is denoted $A^*$ (with some arbitrary fixed tie-breaking rule). The first of these is the linear bandit, where we assume the action set $\cA\subseteq B_1(0)\subseteq \bbR^d$ is contained inside the unit ball. Moreover, the reward $r_t$ is given by
\[
    r_t=\langle \ell^*,A^{(t)}\rangle + z_t
\]
where $z_t$ is conditionally mean zero and $O(1)$-subgaussian conditioned on $(\ell^*,A^{(t)})$. We assume $\ell^*$ is such that 
\[
\bbE[r]=\langle \ell^*,A\rangle\in [-1,1]
\]
for all $A\in\cA$.
This includes $r_t=\langle \ell_t,A^{(t)}\rangle + \tilde z_t$ with $\bbE[\ell_t]=\ell^*$, as well as binary rewards $r_t\in \{-1,1\}$. 
We denote $\theta_i=\langle \ell^*,A_i\rangle$ the expected reward for $A_i\in\cA=(A_1,\dots,A_{|\cA|})$.

The second model we consider is the combinatorial semibandit. Here the action set $\cA\subseteq 2^{[d]}$ is a family of subsets; we call $A\in\cA$ an \emph{action} and $a_i\in [d]$ an \emph{atom}. After playing $A^{(t)}\in\cA$, the player receives a vector of reward feedback $(r_{a})_{a\in A^{(t)}}\in \{0,1\}^{|A^{(t)}|}$ and gains their entrywise sum as a total reward. Each atom $a_i$ gives reward independently with probability $\theta_i$, and following \cite{hu2022incentivizing} we assume that the $\theta_i$ are jointly independent under $\mu$. We let $\theta_A=\sum_{i\in A}\theta_i$ for each $A\in\cA$. We let $n_t(i)$ denote the number of times atom $a_i$ has been sampled prior to time $t$, and let $\hat p_n(i)$ be the empirical average reward of arm $i$ from its first $n>0$ samples. For each $j\in [d]$ we set $\cA_j$ be the subset of $\cA$ consisting of $A$ containing $a_j$, and $\cA_{-j}=\cA\backslash \cA_{j}$.

Thompson sampling is a Bayesian bandit algorithm, defined from an initial prior $\mu$ over $\ell^*$. Let $\mathcal F_t$ denote the observed history strictly before time $t$ and set $\bbE^t[\cdot]=\bbE[\cdot|\mathcal F_t]$, $\bbP^t[\cdot]=\bbP[\cdot|\mathcal F_t]$. Thompson sampling at time $t$ draws an arm independently at random from the time-$t$ distribution of $A^*$, so that 
\[
    \bbP^t[A^{(t)}=A]=\bbP^t[A^*=A].
\]
Given a fixed prior $\mu$, we say a bandit algorithm is BIC if it satisfies \eqref{eq:BIC-def}. More leniently, we say it is $\eps$-BIC if for each $t\in [T]$:
\begin{equation}
\label{eq:eps-BIC}
\begin{aligned}
    &\bbE[\mu(A)-\mu(A')~|~A^{(t)}=A]\geq -\eps,
    \\
    &\quad\forall A,A'\in \cA\text{ such that}~\bbP[A^{(t)}=A]>0.
\end{aligned}
\end{equation}
We have mentioned regret guarantees for Thompson sampling with linear contexts. For the combinatorial semibandit, \cite{TSfirstorder} shows among other things that Thompson sampling attains the optimal $\wt O(d\sqrt{T})$ regret.

We say a mean-zero scalar random variable $X$ is $C$-subgaussian if for all $t\in\bbR$,
\begin{equation}
\label{eq:subgaussian}
    \bbE[e^{tX}]\leq e^{C^2t^2/2}.
\end{equation}
The smallest $C$ such that \eqref{eq:subgaussian} holds is the subgaussian norm of $X$. We say a mean-zero random vector $\vec X\in\bbR^d$ is $C$-subgaussian if $\langle \vec X,\vec v\rangle$ is $C$-subgaussian for all $\vec v\in\bbR^d$ of norm $\|\vec v\|\leq 1$.

We use $\preceq$ to denote the positive semidefinite partial order on symmetric matrices, i.e. $M_1\preceq M_2$ if and only if $M_2-M_1$ is positive semidefinite.

\subsection{Bayesian Chernoff Bounds}

We use the following posterior contraction lemma from \cite{sellke-slivkins}. When $\theta$ is taken to be an empirical average, it yields a Bayesian version of the classical Chernoff bound. The statement below hides some technical specifications for readability, but all quantities lie in $\mathbb R^m$ for some $m$, all functions are Borel measurable, and all probability measures are defined on the Borel sigma algebra.

\begin{lemma}
\label{lem:freqguarantee}
Let $\xi\in\bbR^d$ be an unknown parameter and $\gamma$ an observed signal with distribution depending on $\xi$. Suppose there exists an estimator $\theta=\theta(\gamma)\in\bbR^d$ for $\xi$ depending only on this signal, which satisfies for some deterministic $\eps,\delta>0$ the concentration inequality
\begin{equation}
\label{eq:chernoff-hypothesis}
    \bbP\big[ \|\theta-\xi\|\geq \eps~|~\xi \big]\leq\delta\quad \forall \xi.
\end{equation}

Further, let $\xi\sim \mu$ be generated according to a prior distribution $\mu$, and let $\hat \xi$ be a sample from the posterior distribution $\hat \mu=\hat \mu(\gamma)$ for $\xi$ conditioned on the observation $\gamma$. Then
\begin{equation}
\label{eq:chernoff-conclusion}
    \bbE^{\xi\sim \mu}
    \big[
    \bbP^{\hat\xi\sim \hat \mu}
    \big[ \|\hat\xi-\xi\|\geq 2\eps
    \big]
    \big]
    \leq 2\delta.
\end{equation}
\end{lemma}

\begin{proof}
    The pairs $(\xi,\gamma)$ and $(\hat\xi,\gamma)$ are identically distributed; therefore $(\xi,\theta)$ and $(\hat\xi,\theta)$ are as well. The result now follows by the triangle inequality.
\end{proof}

\begin{lemma}[{\cite[Lemma A.13]{sellke-slivkins}}]
\label{lem:subgaussian-tail}
    Suppose the scalar random variable $X$ is mean zero and $O(1)$-subgaussian and the event $E$ has $\bbP[E]\leq \delta$. Then 
    \[
    \mathbb E[|X\cdot 1_E|]\leq O\big(\delta\sqrt{\log(1/\delta)}\big).
    \]
\end{lemma}

Though the statement of Lemma~\ref{lem:freqguarantee} is abstract, our uses of it will be very concrete. Namely $\xi$ will be the unknown mean reward and $\gamma$ the actions and rewards up to some time $t$. The estimate $\theta$ will be obtained simply by an empirical average or linear regression. Then \eqref{eq:chernoff-hypothesis} amounts to Hoeffding's inequality or a multivariate analog. The conclusion \eqref{eq:chernoff-conclusion} for posterior samples will be of great use in the analysis of Thompson sampling.

\section{Linear Bandit}
\label{sec:linear}

In \cite{sellke-slivkins}, it was shown that Thompson sampling is BIC once a mild amount of initial data has been collected almost surely at a fixed time. Here we show a qualitatively similar result for the linear bandit. The notion of ``amount'' of initial data we adopt is that the action vectors $A^{(1)},\dots,A^{(t)}$ taken so far satisfy the spectral condition
\begin{equation}
\label{eq:spectral-explore}
    \sum_{s=1}^t (A^{(s)})^{\otimes 2}\succeq \gamma I_d.
\end{equation}
If $\cA=(e_1,\dots,e_d)$ forms an orthonormal basis (as in the multi-armed bandit setting), this simply means that each action was sampled at least $\gamma$ times.
We say that $\gamma$-spectral exploration has occurred at time $t$ if \eqref{eq:spectral-explore} holds. 
Note that if one is willing to simply purchase initial samples, then $\gamma d$ samples are typically needed to achieve $\gamma$-spectral exploration.
\msedit{The next lemma follows from Lemma~\ref{lem:subgaussian-tail} and is proved in the appendix.}
We note that the factor of $\sqrt{d \log(t)}$ comes from the adaptivity of the exploration, as explained in \cite[Exercise 20.2]{LS19bandit-book}.

\begin{lemma}
\label{eq:subgaussian-book}
Suppose $\gamma$-spectral exploration has occurred almost surely at some (deterministic) time $t$. Then the random vector $\ell^*-\bbE^t[\ell^*]$ has zero mean and is $O(\sqrt{d\log(t)/\gamma})$-subgaussian.
\end{lemma}


Our results will hold for $\eps$-separated action sets $\cA$ as defined below. We view this as a generic assumption, e.g. it holds with high probability when $\cA\subseteq\bbS^{d-1}$ consists of $e^{O(d)}$ randomly chosen points. Moreover it is common to discretize infinite action sets for the purposes of analysis. From a technical point of view, requiring discrete $\cA$ ensures that the event conditioned on has non-negligible probability.

\begin{definition}
    The set $\cA\subseteq \bbS^{d-1}$ is said to be $\eps$-separated if $\|A_1-A_2\|\geq \eps$ for any \textbf{distinct} $A_1, A_2\in\cA$.
\end{definition}

\subsection{Convex Geometry}

\msedit{
We start with two definitions from convex geometry.
}

\begin{definition} 
For $\cK\subseteq \bbR^d$ a compact convex set with non-empty interior, let $\cU(\cK)$ denote the uniform measure on $\cK$. If $v\in\bbR^d$ we let
\[
\width_{v}(\cK)
=
\max_{\ell\in\cK}\,\langle \ell,v\rangle
-
\min_{\ell\in\cK}\,\langle \ell,v\rangle.
\]
\end{definition}

\begin{definition}
    The convex set $\cK\subseteq \bbR^d$ is $r$-regular if $B_r(0)\subseteq \cK\subseteq B_1(0)$.
\end{definition}

\msedit{The next lemma lets us connect convexity conditions to the BIC property. Its proof follows from standard tools and is given in the Appendix. 

\begin{lemma}
\label{lem:width}
For any convex $\cK\subseteq \bbR^d$,
\[
    \bbE^{\ell\sim\cU(\cK)}|\langle \ell,v\rangle|
    \geq
    \Omega(\width_{v}(\cK)/d)
    .
\]
\end{lemma}
}

Our key technical estimate is below. It requires $\ell^*$ to have uniform prior $\cU(\cK)$ over a regular convex body $\cK$. Note that since the left-hand side of \eqref{eq:main-technical-linear-bandit} is $1$-homogeneous, it is also fine for $\ell^*$ to be drawn from a \textbf{centered} Gaussian with covariance $\Sigma$ satisfying $\lambda I_d\preceq \Sigma\preceq \Lambda I_d$.

\begin{corollary}
\label{cor:linear-margin}
    Suppose that $\cA\subseteq\bbS^{d-1}$ is $\eps$-separated, the convex set $\cK\subseteq \bbR^d$ is $r$-regular, and $\ell^*\sim\mu=\cU(\cK)$. Then for each $A_i,A_j\in\cA$, we have
    \begin{align}
    \label{eq:deltai-LB}
    \bbP[A_i=A^*(\ell^*)]
    &\geq
    \lt(\msedit{r\eps/4}\rt)^d,
    \\
    \label{eq:main-technical-linear-bandit}
    \bbE[\langle \ell^*,A_i-A_j\rangle|A^*=A_i]
    &\geq
    \Omega\lt(\frac{r\eps \|A_i-A_j\|}{d}\rt).
    \end{align}
\end{corollary}
\vspace{-0.6cm}
\msedit{
\begin{proof}
    Fix $i$ and define the convex set $S_i=\{\ell\in B_1~:~ A_i = A^*(\ell)\}$.
    To prove \eqref{eq:deltai-LB}, we first show 
    \begin{equation}
    \label{eq:Siepsball}
    B_{\eps/2}(A_i)\subseteq S_i.
    \end{equation}
    Indeed suppose $\ell\in \bbR^d$ satisfies $\|\ell-A_i\|\leq \frac{\eps}{2}$. Then for any $a\in\cA$ different from $A_i$, 
    \begin{align*}
    \langle \ell,A_i-A\rangle
    &=
    \langle A_i,A_i\rangle -\langle A_i,a\rangle + \langle \ell-A_i,A_i-A\rangle
    \\
    &\geq
    1-\langle A_i,a\rangle - \frac{\eps}{2}\|A_i-A\|_2
    \\
    &=
    \|A_i-A\|_2^2/2 -\frac{\eps}{2} \|A_i-A\|_2
    \geq 0.
    \end{align*}
    In the last step we used the fact that $t^2-\eps t\geq 0$ for $t\geq \eps$ (recall that $\cA$ is $\eps$-separated). We conclude that $\ell\in S_i$, hence establishing \eqref{eq:Siepsball}.
    By $r$-regularity of $\cK$ and homogeneity of $S_i$ it follows that 
    \[
    B_{r\eps/4}(rA_i/2)\subseteq \cK\cap S_i.
    \]
    Since $\cK\subseteq B_1(0)$ we find $\Vol(\cK\cap S_i)/\Vol(\cK)\geq (r\eps/4)^d$, implying \eqref{eq:deltai-LB}.
    Next, $r$-regularity of $\cK$ and \eqref{eq:Siepsball} imply 
    \begin{align*}
        \width_{v}(S_i\cap\cK)
        &\geq 
        r\cdot \width_v(S_i)
        \\
        &\geq
        \Omega(r\eps \|v \|)
    \end{align*}
    for any $v\in\bbR^d$. Hence by Lemma~\ref{lem:width}, for any fixed $v$
    \[
    \bbE^{\ell\sim \cU(S_i\cap\cK)}\big| \langle \ell,v\rangle\big| 
    \geq
    \Omega(r\eps \|v\|/d).
    \]
    Setting $v=A_i-A_j$, we note that by definition $\langle \ell,A_i-A_j\rangle \geq 0$ for all $\ell\in S_i$. Hence
    \[
    \bbE^{\ell\sim \cU(S_i\cap\cK)} \langle \ell,A_i-A_j\rangle  
    \geq
    \Omega(r\eps \|A_i-A_j\|/d)
    \]
    which is equivalent to \eqref{eq:main-technical-linear-bandit}.
\end{proof}
}

\subsection{Main Result for Linear Bandit}

We now show our main result on Thompson sampling for the linear bandit, namely that $\gamma$-spectral exploration suffices for Thompson sampling to be BIC when $\gamma \geq \frac{C  d^3 \log(\msedit{4/r\eps})}{r^2\eps^2}$ for $\eps$-separated action sets. This roughly means $\wt O\lt(d^4/\eps^2\rt)$ ``well-dispersed'' samples are required for Theorem~\ref{thm:main-linear-bandit} to apply.

\begin{remark}
In fact a more general statement is true. Given an $\eps_t$-discretization $\cA^{(t)}$ of $\cA$, let us recommend an action $A^{(t)}\in\cA^{(t)}$ by performing Thompson sampling on $\cA^{(t)}$. Then the proof of Theorem~\ref{thm:main-linear-bandit} shows that a rational user will always choose an action $\tilde A^{(t)}\in \cA$ within distance $\eps_t$ of $A^{(t)}$, since all other actions are inferior to $A^{(t)}$. 
\end{remark}

\begin{theorem}
\label{thm:main-linear-bandit}
    There exists an absolute constant $C$ such that the following holds. Suppose that $\cA\subseteq\bbS^{d-1}$ is $\eps$-separated, the convex set $\cK\subseteq \bbR^d$ is $r$-regular, and $\ell^*\sim\mu=\cU(\cK)$. If \eqref{eq:spectral-explore} holds almost surely at time $t$ for 
    \[\gamma \geq \frac{C  d^4 \log(t) \log(\msedit{4/r\eps})}{r^2\eps^2},
    \]
    then Thompson sampling for the linear bandit is BIC at time $t$. 
\end{theorem}

\begin{proof}
As in the proof of Theorem 4.1 in \cite{sellke-slivkins}, it suffices to show that 
\[
    \bbE[1_{A^*=A_i}\cdot \bbE^{t}[\langle \ell^*,A_i-A_j\rangle]]\geq 0
\]
for any action $A_j\neq A_i$. Defining $\delta_i=
    \bbP[A_i=A^*(\ell^*)]
    \stackrel{\eqref{eq:deltai-LB}}{\geq}
    \lt(\msedit{\frac{r\eps}{4}}\rt)^d$, we have
\begin{align*}
    \bbE\big[1_{A^*=A_i}\cdot \langle \ell^*,A_i-A_j\rangle\big]
    &=
     \bbE\big[1_{A^*=A_i}\cdot (\langle \ell^*,A_i-A_j\rangle)_+\big]
     \\
     &=
     \delta_i\cdot
     \bbE\big[(\langle \ell^*,A_i-A_j\rangle)_+~|~A^*=A_i\big]
     \\
     &\stackrel{Lem.~\ref{lem:width}}{\geq}
     \Omega(\delta_i r\eps \|A_i-A_j\|/d).
\end{align*}
Recall that $\theta_i=\langle \ell^*,A_i\rangle$ denotes the expected reward for action $A_i$. It remains to upper-bound
\[
    \bbE\lt[
    1_{A^*=A_i}
    \cdot
    \lt|
    \bbE^{t}
    [\theta_i-\theta_j]
    -(\theta_i-\theta_j)
    \rt|
    \rt].
\]
By Lemma~\ref{eq:subgaussian-book}, $\frac{\gamma^{1/2}\big(\bbE^{t}[\theta_i-\theta_j]-(\theta_i-\theta_j)\big)}{\|A_i-A_j\|\sqrt{d\log(t)}}$ is centered and $O(1)$ subgaussian (sans conditioning). Using Lemma~\ref{lem:subgaussian-tail}, it follows that
\begin{align*}
    \bbE\lt[
    1_{A^*=A_i}
    \cdot
    \lt|
    \bbE^{t}
    [\theta_i-\theta_j]
    -(\theta_i-\theta_j)
    \rt|
    \rt]
    &
    \leq
    O\lt(\delta_i
    \|A_i-A_j\|
    \sqrt{\frac{d\log(t)\log(1/\delta_i)}{\gamma}}\rt).
\end{align*}
Since we assumed
\[
    \gamma\geq \frac{C d^4 \log(t)\log(4/r\eps)}{r^2\eps^2}
\]
for a large enough absolute constant $C$, we finish the proof via:
\begin{align*}
    \bbE\lt[
    1_{A^*=A_i}
    \cdot
    \lt|
    \bbE^{t}
    [\theta_i-\theta_j]
    -(\theta_i-\theta_j)
    \rt|
    \rt]
    &\leq
    O\lt(\delta_i
    \|A_i-A_j\|
    \sqrt{\frac{d\log(t)\log(1/\delta_i)}{\gamma}}\rt)
    \\
    &\leq
    \Omega\lt(\frac{\delta_i r\eps \|A_i-A_j\|}{d}\rt)
    \\
    &\leq
    \bbE\big[1_{A^*=A_i}\cdot \langle \ell^*,A_i-A_j\rangle\big]
\end{align*}
In the second step we used $\log(1/\delta_i)\leq d\log(4/r\eps)$ which follows from \eqref{eq:deltai-LB}.
\end{proof}

\begin{remark}
    The action set $\cA$ only enters Theorem~\ref{thm:main-linear-bandit} at time $t$. Thus it holds even if the preceding $\gamma$-spectral exploration used actions not in $\cA$. In particular Theorem~\ref{thm:main-linear-bandit} extends to the case that $\cA^{(t)}$ changes over time as long as $\cA^{(t)}$ is $\eps$-separated. 
    
    For example, suppose that each $\cA^{(t)}$ consists of $e^{cd}$ i.i.d. uniform vectors on the unit sphere for a small constant $c$, but the sets $\cA^{(t)}$ may be arbitrarily dependent. (For example, we might replace the actions gradually.) Each $\cA^{(t)}$ is $\eps$-separated with probability $1-e^{-\Omega(d)}$. Then given a $\gamma$-spectral warm start, the algorithm \emph{use Thompson sampling if $\eps$-separation holds, else act greedily} will be BIC and incur expected regret
    \[
    O\lt(d\sqrt{T}\log T + e^{-\Omega(d)}T\rt).
    \] 
    This follows by Proposition 3 of \cite{Russo-MathOR-14}, which allows for changing action sets $\cA^{(t)}$.
\end{remark}

\subsection{Extension to Generalized Linear Bandit}

Here we give an analogous proof for the  generalized linear bandit, which was introduced in \cite{GeneralizedLinear-nips10}. In this model, the expected reward is 
\[
    \bbE[r_t|A^{(t)}]=\chi(\langle A^{(t)},\ell^*\rangle)
\]
for a known strictly increasing \textbf{link function} $\chi:\bbR\to\bbR$ and unknown vector $\ell^*$. A notable example is the logistic bandit where $\chi(x)=\frac{e^x}{1+e^x}$. Recalling that we assumed $\langle \ell^*,x\rangle\in [-1,1]$ for all $\ell^*\in\supp(\mu)$ and $x\in A$, we follow much of the theoretical literature on this model by fixing the non-negative constants
\[
    M_{\chi}=\sup_{x\in [-1,1]} \chi'(x);
    \quad\quad
    m_{\chi}=\inf_{x\in [-1,1]} \chi'(x).
\]
(See e.g. \cite{neu2022lifting} for an analysis that avoids such dependence.)
 Our analysis of the linear bandit extends to this setting assuming \emph{deterministic} initial exploration; the proof is given in the Appendix.
\begin{theorem}
\label{thm:main-GLM}
    Fix a link function $\mu:\bbR\to\bbR$. Suppose that $\cA\subseteq\bbS^{d-1}$ is $\eps$-separated, the convex set $\cK\subseteq \bbR^d$ is $r$-regular, and $\ell^*\sim\mu=\cU(\cK)$, where $r,\eps,m_{\chi}\geq e^{-d}$.
    Given deterministic initial exploration until time $t$ under which \eqref{eq:spectral-explore} holds for
    \[
    \gamma\geq \frac{C M_{\chi}^2 d^3 \log(4/r\eps)}{r^2 m_{\chi}^4 \eps^2},
    \]
    Thompson sampling for the generalized linear bandit is BIC at time $t$.
\end{theorem}

\subsection{Two Counterexamples}

Theorem 4.7 and Lemma 4.9 of \cite{sellke-slivkins} show that for the multi-armed bandit with an independent product prior over arms, if Thompson sampling is BIC at time $s$ then it is also BIC at any future time $t>s$. (In fact, this holds even if an arbitrary and possibly non-BIC bandit algorithm is used between times $s$ and $t$.) This is an appealing general ``stability'' property for Thompson sampling. In particular, the special case $s=0$ implies that if all arms have the same prior mean reward, then Thompson sampling is BIC with no initial exploration phase.

We give an explicit counterexample (with proof in the appendix) showing that even the $s=0$ case fails for the linear bandit when $(s,t)=(1,2)$ with $d=2$. Precisely:

\begin{proposition}
\label{prop:counterexample}
    Let $\mu\sim \cN(0,I_2)$ and 
    \[
    (A_1,A_2,A_3)=\big( (1,0),~(-1,0),~(1.8,0.6)\big).
\] 
If Thompson sampling recommends action $A_1$ at time $t=2$, this recommendation is \textbf{not} BIC:
    \begin{equation}
    \label{eq:counterexample}
    \bbE[\theta_3~|~A^{(2)}=A_1]>\bbE[\theta_1~|~A^{(2)}=A_1].
    \end{equation}
\end{proposition}
Note that although $A_3$ has different Euclidean norm from $A_1$ and $A_2$, this set can be put onto the unit circle by a linear change of basis (which just changes $\mu$ to a non-spherical Gaussian) in accordance with Section~\ref{sec:linear}.

For our second counterexample, suppose the action set $\cA$ is either the polytope 
\[
    \cP\equiv\lt\{(x_1,\dots,x_d)~:~10|x_d|+2\sqrt{d}\max_{i<d}|x_i|\leq 1\rt\}\subseteq B_1
\]
or its set of extreme points $\cE(\cP)$. We show that for a suitable prior $\mu$, \emph{any} initial exploration algorithm which is BIC or even $\eps$-BIC can be made to require $e^{\Omega(d)}$ timesteps to explore in the $d$-th coordinate.

Note that for $\cP$ as above, all extreme points satisfy $x_d=0$ except for $\pm\hat A = \lt(0,0,\dots,0,\pm \frac{1}{10}\rt)$, and so $\gamma$-spectral exploration for any $\gamma>0$ requires exploring these actions. Our prior distribution $\ell^*\sim\mu$ will be uniform over the ``biased'' convex body $\cK/\sqrt{d}$ for
\[
    \cK
    \equiv
    [-0.5,1]^{d-1} \times [-1,1].
\]

\begin{proposition}
\label{prop:linear-bandit-LB}
    With $\cA,\mu$ as above for $d\geq 100$, suppose that $\bbA$ is either an $0.01$-BIC algorithm on action set $\cE(\cP)$, or a BIC algorithm on $\cP)$. 
    If at least one of $\hat A$ or $-\hat A$ is almost surely explored by $\bbA$ before time $T$, then $T\geq \exp(\Omega(d))$.
\end{proposition}

\section{Initial Exploration for Combinatorial Semibandit}

It was found in \cite{hu2022incentivizing} that for incentivized exploration, the semibandit problem is well-suited for generalizing \cite{sellke-slivkins}. Indeed one can still assume the atoms $i\in [d]$ give rewards independently with probabilities $\theta_1,\dots,\theta_d$ which are independent under the prior $\mu$; this allowed them to generalize several proofs using independence of atoms rather than full actions. However \cite{hu2022incentivizing} did not extend the initial exploration scheme of \cite{sellke-slivkins}, instead proposing several algorithms with exponential sample complexity. Here we find the natural extension of \cite{sellke-slivkins}, again connecting the sample complexity of BIC initial exploration to the value of a two-player zero sum game.

Following \cite{hu2022incentivizing}, we make some non-degeneracy assumptions. Namely we require that all atom rewards are almost surely in $[0,1]$ and their distributions satisfy:
\begin{enumerate}
    \item $\bbE[\theta_{\ell}]\geq\tau$.
    \item $\Var[\theta_{\ell}]\geq\sigma^2$.
    \item $\bbP[\theta_{\ell}<x]\geq \exp(-x^{-\alpha})$.
\end{enumerate}
for fixed constants $\tau,\sigma,\alpha>0$.

Under these conditions, we give an initial exploration algorithm in the style of \cite{sellke-slivkins}. We will explore the arms in a deterministic order which is defined via the lemmas below. Throughout, we say that an atom has been $N$-explored if it has been sampled at least $N$ times. Moreover we define the event
\[
    \ZEROS_{<j,N} =
    \lt\{
    \hat p_{N}(i)=0~\forall i<j
    \rt\}
\]
that each $i<j$ receives no reward in its first $N$ samples,
and also set
\begin{equation}
\label{eq:eps-j}
    \eps_j=\bbP[\ZEROS_{<j,N}].
\end{equation}
It follows from the assumptions above that for all $j\in [d]$,
\[
    \eps_j
    \geq
    (1-p)^{dN}\exp(-dp^{-\alpha}).
\]
In particular setting $p=\frac{1}{dN}$ we find:
\begin{equation}
\label{eq:eps-LB}
    \log(1/\eps_j)
    \leq
    \Omega\lt(d^{1+\alpha}N^{\alpha}\rt).
\end{equation}

\begin{lemma}
\label{lem:get-new-action}
    Let $S_{j-1}=\{a_1,\dots,a_{j-1}\}\subseteq [d]$ and assume that all atoms in $S_{j-1}$ have been $N$-explored almost surely at time $T_{j-1}$, for 
    \begin{equation}
    \label{eq:N-LB}
        N\geq \big(20d/\tau\big)^{1+\alpha}\log(20d/\tau).
    \end{equation}
    Then conditioned on the event $\ZEROS_{<j,N}$, the action with highest posterior mean contains an action in $[d]\backslash S_{j-1}$. More generally, let
    \[
    x_j=\max\big(1_{\ZEROS_{<j,N}},b_j\big)
    \]
    for $y_j=b_j\sim \Ber\lt(q_j\rt)$ an independent Bernoulli random variable. Then the conclusion remains true conditionally on the event $x_j=1$ whenever $q_j\leq \eps_j=\bbP[\ZEROS_{<j,N}]$.
\end{lemma}

\begin{proof}
    We claim that for any actions $A,A'\in\cA$ with $A\subseteq S_{j-1}$ and $A'\not\subseteq S_{j-1}$, we have 
    \[
    \bbE[\theta_A~|~\ZEROS_{<j,N}]\leq\bbE[\theta_{A'}~|~\ZEROS_{<j,N}]/2.
    \]
    Thanks to the factor of two, this implies that
    \[
    \arg\max_{A'\in \cA}\bbE[\theta_{A'}~|~x_j=1] \not\subseteq S_{j-1}
    \]
    for $q_j\leq\eps_j$. (Note that the left-hand side is a deterministic subset of $[d]$ because $\{x_j=1\}$ is just a single event; we do not condition on all the information of $\cF_{T_{j-1}}$.)

    First, if $a\notin S_{j-1}$, then $\bbE[\theta_a~|~\ZEROS_{<j,N}]=\bbE[\theta_a]\geq \tau$ by assumption. It remains to show that 
    \begin{equation}
    \label{eq:bad-so-far-bound-v0}
    \bbE[\theta_a~|~\ZEROS_{<j,N}]\stackrel{?}{<} \frac{\tau}{2d},\quad \forall~a\in S_{j-1}.
    \end{equation}
    We will show that in fact
    \begin{equation}
    \label{eq:bad-so-far-bound}
    \bbP\lt[\theta_a\geq\frac{\tau}{5d}|~\ZEROS_{<j,N}\rt]
    \leq
    \frac{\tau}{5d}
    \end{equation}
    which implies \eqref{eq:bad-so-far-bound-v0} since $\theta_a\leq 1$ almost surely.
    To see \eqref{eq:bad-so-far-bound} we use likelihood ratios to bound probabilities. Note that for any $a\in S_{j-1}$
    \begin{align*}
    \bbP\lt[\theta_a\geq \frac{\tau}{3d}~|~\ZEROS_{<j,N}\rt]
    \leq
    \frac{\bbP[\theta_a\geq \frac{\tau}{3d}~|~\ZEROS_{<j,N}]}
    {\bbP[\theta_a\leq \frac{\tau}{6d}~|~\ZEROS_{<j,N}]}
    &\leq
    \lt(\frac{1-\frac{\tau}{3d}}{1-\frac{\tau}{6d}}\rt)^N
    \frac{\bbP[\theta_a\geq \frac{\tau}{3d}]}
    {\bbP[\theta_a\leq \frac{\tau}{6d}]}
    \\
    &\leq \frac{e^{-(N\tau/10d)}}{\bbP[\theta_a\leq \frac{\tau}{6d}]}
    \\
    &\leq
    \exp\lt(\big(10d/\tau\big)^{\alpha}-\frac{N\tau}{10d}\rt).
    \end{align*}
    Since $N$ is at least
    \[
    \big(20d/\tau\big)^{1+\alpha}\log(20d/\tau)
    \geq 2\big(10d/\tau\big)^{1+\alpha}   
    \log(10d/\tau)
    \]
    the upper bound on $\bbP\lt[\theta_a\geq \frac{\tau}{3d}~|~\ZEROS_{<j,N}\rt]$ is at most 
    \[
        \exp(-\big(10d/\tau\big)^{\alpha}\log(10d/\tau))\leq (\tau/10d)^{\big(10d/\tau\big)^{\alpha}}\leq \frac{\tau}{3d}.
    \]
    This concludes the proof.
\end{proof}

\subsection{Algorithm~\ref{alg:main}}

In light of Lemma~\ref{lem:get-new-action}, we can assume without loss of generality that the atoms are ordered such that $a_j$ is BIC conditioned on the event $1=\max(1_{\ZEROS_{<j,N}},b_j)$, for $b_j\sim \Ber(q_j)$ as in the lemma. Then we will explore the atoms in increasing order. For each $j\in [d]$, our high-level strategy is as follows:
\begin{enumerate}
    \item On the event $\ZEROS_{<j,N}$, obtain $N$ samples of $a_j$.
    \item Exponentially grow the probability to have $N$-explored $a_j$.
\end{enumerate}
Similarly to \cite{sellke-slivkins}, the main insight is that it is possible to achieve an exponential growth rate. This enables sample efficiency even when the probability from the first phase is exponentially small. We will assume $N$ is large enough to satisfy \eqref{eq:N-LB}; if not, one simply increases $N$ and obtains some additional samples in using the algorithm.

Pseudo-code for our Algorithm~\ref{alg:main} is given below; note that when we say to ``exploit'' given a signal, we simply mean the planner should recommend the greedy action given some signal he generated; this is BIC by definition. The policies $\pi_j$ are important for the exponential growth strategy and are non-obvious to construct. However much of the algorithm can be understood more easily. First, in the $j$-th iteration of the outer loop, we collect $N$ samples of atom $a_j$ almost surely, this having been accomplished for all $i<j$ in the previous iterations. The starting point for each loop is to sample $a_j$ on the event $\ZEROS_{<j,N}$ that all atoms $i<j$ received no reward during their first $N$ samples. While $\ZEROS_{<j,N}$ has tiny probability (denoted $\eps_j$), we grow the probability to sample $a_j$ exponentially using $\pi_j$. The sample complexity is largely dictated by the optimal such exponential growth rate, which is the minimax value of a two-player zero sum game; therefore we take $\pi_j$ to be an approximately optimal strategy for said game.

\begin{algorithm2e}[t]
\caption{\MainALG}
\label{alg:main}
\SetAlgoLined\DontPrintSemicolon
\textbf{Parameters:}
    Number of desired samples $N$ satisfying \eqref{eq:N-LB},
    Minimax parameter \msedit{$\lambda=\ulambda/2d$}.
\\
\textbf{Given:}
recommendation policies $\pi_2 \cdots \pi_d$ for \PaddedPhase.
\\
\textbf{Initialize:} \ExplorePhase[1]
\\
\For{each arm $j=a_{1},\dots,a_{d}$}{ \label{line:algloop}
    \tcp*[l]{each arm $a_{i}$ for $i<j$ has been sampled at least $N$ times}
    \vspace{1mm}
    Generate $y_j=b_j\sim \Ber(\eps_j)$ and set
    $x_j=\max(1_{\ZEROS_{<j,N}},b_j)$.
    \\
    \ExploitPhase[N], conditionally on the signal $x_j$. 
    \label{line:alg-first-step}
    \\
     \tcp*[l]{If $x_j=1$, then this explores $a_j$}
    $p_j\leftarrow \eps_j$.
    \\
    \tcp*[l]{main loop: exponentially grow the exploration probability}

    \While{$p_j< 1$}{
    \tcp*[l]{Try to increase $b_j$ to $1$, and maintain $y_j=\bbP[b_j=1]$.}
    Generate $b_j\sim \Ber\big(\min(1,\frac{p_j\lambda}{1-p_j})\big)$ and set $z_j=\max(b_j,y_j)$.
    \\
    \uIf{$z_j=1$}{
        \PaddedPhase: use policy $\pi_j$ for $N$ time-steps \label{line:padded2}
        }
    Otherwise, exploit for $N$ time-steps
    \\
    Update $p_j\leftarrow \min\lt(1,\, p_j\, (1+\padding)\rt)$.
} 
} 
\end{algorithm2e}

\subsection{The Policy $\pi_j$}

The policy $\pi_j$ is constructed via a $j$-recommendation game, a two-player zero sum game which generalizes the one considered in \cite{sellke-slivkins}. In particular, the optimal exponential growth rate of exploration is dictated by its minimax value in the perfect information (infinite-sample) setting. In the Appendix, we treat finite-sample approximations to it which are required for constructing $\pi_j$.

\begin{definition}
The \emph{infinite sample $j$-recommendation game} is a two-player zero sum game played between a \emph{planner} and an \emph{agent}. The players share a common independent prior over the true mean rewards $(\theta_i)_{i\leq j}$, and the planner gains access to the values of $(\theta_i)_{i\leq j}$. The planner then either: 1) \emph{Recommends an arm $A_j\in \cA_j$ containing $a_j$}, or 2) \emph{Does nothing}.

In the first case, the agent observes $A_j\in \cA_j$ and chooses a response $A_{-j}$ from a mixed strategy on $\cA_{-j}$ which can depend on $A_j$. The payoff for the planner is zero if the planner did nothing, and otherwise $\theta_{A_j}-\theta_{A_{-j}}$. Let the minimax value of the game for the planner be $\lambda_j\geq 0$ and 
\[
\ulambda\equiv\min_{j\in[d]}\lambda_j.
\]
\end{definition}

It is clear that $\lambda_j\geq 0$ since the first player can always do nothing. In general, the value of this game measures the planner's ability to convince the agent to try arm $j$. As in \cite{sellke-slivkins}, this can be much larger than the expected regret from not being able to play arm $j$. The importance of the value $\ulambda$ is explained by the following lemma, which ensures that with some care, it is possible to increase the $j$-exploration probability by a factor $1+\Omega(\ulambda)$ at each stage. 

\begin{lemma}
\label{lem:BIC-growth}
Let \msedit{$\lambda=\ulambda/2d$}, so in particular $\lambda\leq \lambda_j/2$. For $N\geq \Omega\lt(\frac{d^2 \log|\cA|}{\lambda^2}\rt)$, let $z$ be a signal such that conditioned on $z=1$, with probability at least $\frac{1}{1+\lambda}$ the first $j$ arms have all been $N$-explored. Moreover, suppose that the event $z=1$ is independent of the true mean reward vector $\vec \mu$. Then there is a BIC policy which always plays an action in $\cA_j$ when $z=1$.
\end{lemma}

\begin{theorem}
    Algorithm~\ref{alg:main} is BIC and almost surely obtains $N$ samples of each atom within $T$ timesteps for
    \begin{equation}
    \label{eq:alg-sample-comp-bound}
    T=
    O\lt(\frac{\msedit{d^{3+\alpha}} N^{1+\alpha}}{\ulambda}\rt).
    \end{equation}
\end{theorem}

\begin{proof}
    For the $j$-th of the $d-1$ iterations of the outer loop, in the first step Lemma~\ref{lem:get-new-action} implies that if $x_j=1$, then Line~\ref{line:alg-first-step} explores a new action $a_k$ for $k\geq j$, which is without loss of generality $a_j$.

    For the inner loop, note that we do not consider the event $\ZEROS_{<j,N}$ and only use the ``clean'' data coming from the event $b_j=1$. This prevents the data-dependent trajectory of the algorithm from causing unwanted correlations. In particular Line~\ref{line:padded2} is then BIC by the defining property of $\pi_j$ since the required independence holds. Other steps of the algorithm are BIC by virtue of being exploitation conditionally on some signal.
    
    To complete the analysis, we have $p_j=1$ in the final iteration of the inner loop, and so $b_j=1$ almost surely here. Hence some action in $\cA_i$ is chosen each of these times by definition of $\pi_j$. This holds for each $j$, so $N$ samples of each atom are obtained. For sample complexity, the inner loop requires $O\lt(\frac{\msedit{d}N\log\lt(1/\eps_j\rt)}{\lambda}\rt)$ steps; recalling the bound \eqref{eq:eps-LB} now yields \eqref{eq:alg-sample-comp-bound}. The last line in the algorithm contributes a lower order term $O(dN/\ulambda)$.
\end{proof}


\subsection{Lower Bound}



We now give a sample complexity lower bound for each $a_j$. 
The lower bound is inversely proportional to the minimax value $\olambda_{j}\geq\lambda_{j,\infty}$ of the infinite-sample \textbf{easy $j$-recommendation game} in which all values $(\theta_i)_{i\neq j}$ are known exactly to player $1$.

\begin{proposition}
\label{prop:semibandit-LB}
    For any $j\in [d]$, the sample complexity to almost surely BIC-explore arm $a_j$ is at least $\Omega\lt(\frac{\sigma^2}{d\olambda_j}\rt)-2$.
\end{proposition}

The question of which values $\theta_i$ should be treated as known when proving a lower bound on $a_j$ leads to a mismatch with our upper bound, i.e. we may have $\ulambda<\olambda$ in general. This issue was not present in the multi-armed bandit case of \cite{sellke-slivkins}; without combinatorial action sets, one can actually explore the arms in decreasing order of prior mean. This was justified using the FKG inequality, see Appendix B and Definition 3 therein. However with combinatorial action sets there is no canonical ordering of atoms, which leads to the mismatch above.

\subsection*{Acknowledgement}

Thanks to Anand Kalvit for bringing \cite[Exercise 20.2]{LS19bandit-book} to our attention, and thus correcting an error in the regret bounds for linear bandits. Thanks also to Ben Schiffer for pointing out a number of typos in the second counterexample.

{
\footnotesize
\bibliographystyle{alpha}
\bibliography{bib}
}

\normalsize

\appendix

\section{Proofs for Section~\ref{sec:linear}}

\msedit{
\begin{proof}[Proof of Lemma~\ref{eq:subgaussian-book}]
    The mean zero property just says $\bbE[\bbE^t[\ell^*]]=\bbE[\ell^*]$, which follows by the tower rule for conditional expectations. 
    Theorem 20.5 of \cite{LS19bandit-book} implies that for any unit vector $v\in\bbR^d$, there exists a regularized least squares estimator $\hat\ell^*$ such that $\ell^*-\hat\ell^*$
    is $O(\sqrt{d\log(t)/\gamma})$-subgaussian. From Lemma~\ref{lem:freqguarantee} we find that $\ell^*-\tilde\ell^*$
    is also $O(\sqrt{d\log(t)/\gamma})$-subgaussian, for $\tilde\ell^*$ a posterior sample for $\ell^*$ at time $t$. The result now follows since the sum of two arbitrarily coupled $O(\sqrt{d\log(t)/\gamma})$-subgaussian vectors is still $O(\sqrt{d\log(t)/\gamma})$-subgaussian (by applying the AM-GM inequality to \eqref{eq:subgaussian}).
\end{proof}

\begin{proof}[Proof of Lemma~\ref{lem:width}]
    By Theorem 4.1 of \cite{kannan1995isoperimetric}, 
    \[
    \sqrt{\Var_{\ell\sim\cU(\cK)}\langle \ell,v\rangle}
    \geq
    \Omega(\width_{v}(\cK)/d).
    \]
    Since $\cU(\cK)$ is a log-concave distribution, Theorem 5.1 of
    \cite{guedon2014concentration} implies
    \[
    \bbE^{\ell\sim\cU(\cK)}|\langle \ell,v\rangle|\geq \Omega\lt(\sqrt{\Var_{\ell\sim\cU(\cK)}\langle \ell,v\rangle}\rt)
    .
    \]
    Combining yields the desired estimate.
\end{proof}
}

\subsection{Proofs for Generalized Linear Bandit}

As explained in \cite{GeneralizedLinear-nips10}, the maximum likelihood estimator (MLE) $\hat\ell$ at time $t$ is unique, and is in fact the solution to
\begin{equation}
\label{eq:GLM-MLE}
    \sum_{s=1}^{t}
    \big(R_s-\chi(\langle A_s,\hat\ell\rangle)\big) A_s=0
    .
\end{equation}

The next result essentially shows that deterministic $\gamma$-spectral exploration for $\gamma\geq \frac{C M_{\chi}^2 d^2}{m_{\chi}^4}$ suffices for the estimation error of the MLE to be subgaussian.

\begin{theorem}[{Theorem 1 of \cite{li2017provably}}]
\label{thm:GLM-MLE-subgaussian}
Suppose a deterministic initial exploration sequence yields $\gamma$-spectral exploration for
\[
    \gamma\geq \frac{C M_{\chi}^2}{m_{\chi}^4}\lt(d^2+\log(1/\delta)\rt).
\]
Then the MLE estimate $\hat\ell$ has error $(\hat\ell-\ell^*)$ satisfying for any fixed $v\in\bbR^d$:
\[
    \bbP\lt[
    |\langle \hat\ell-\ell^*,v\rangle|
    \leq
    \frac{\|v\|}{m_{\chi}\sqrt{\gamma}}
    \rt]
    \geq 1-\delta.
\]
\end{theorem}

The following result follows directly from Corollary~\ref{cor:linear-margin}.

\begin{lemma}
\label{lem:GLM-margin}
    Fix a link function $\chi$. Suppose that $\cA\subseteq\bbS^{d-1}$ is $\eps$-separated, the convex set $\cK\subseteq \bbR^d$ is $r$-regular, and $\ell^*\sim\mu=\cU(\cK)$. Then for each $A_i,A_j\in\cA$, we have
    \begin{align*}
    \bbP[A_i=A^*(\ell^*)]
    &\geq
    \lt(\frac{\eps}{2}\rt)^d,
    \\
    \bbE\big[\chi(\langle \ell^*,A_i\rangle)-\chi(\langle \ell^*,A_j\rangle)~|~A^*=A_i\big]
    & \geq
    \Omega(r\eps \|A_i-A_j\| m_{\chi}/d).
    \end{align*}
\end{lemma}

We now extend Theorem~\ref{thm:main-linear-bandit} to the generalized linear bandit model.

\begin{proof}[Proof of Theorem~\ref{thm:main-GLM}]
As before it suffices to show that 
\[
    \bbE\Big[1_{A^*=A_i}\cdot \bbE^{t}\big[\chi(\langle \ell^*,A_i\rangle)-\chi(\langle \ell^*,A_j\rangle)\big]\Big]\geq 0
\]
for any action $A_j\neq A_i$. Defining $\delta_i=
    \bbP[A_i=A^*(\ell^*)]
    \stackrel{\eqref{eq:deltai-LB}}{\geq}
    \lt(\msedit{\frac{r\eps}{4}}\rt)^d$, we have
\begin{align*}
    \bbE\Big[1_{A^*=A_i}\cdot \big[\chi(\langle \ell^*,A_i\rangle)-\chi(\langle \ell^*,A_j\rangle)\big]\Big]
    &=
     \bbE\big[1_{A^*=A_i}\cdot \big[\chi(\langle \ell^*,A_i\rangle)-\chi(\langle \ell^*,A_j\rangle)\big]_+\big]
     \\
     &=
     \delta_i\cdot
     \bbE\big[\big[\chi(\langle \ell^*,A_i\rangle)-\chi(\langle \ell^*,A_j\rangle)\big]_+~|~A^*=A_i\big]
     \\
     &\stackrel{Lem.~\ref{lem:width}}{\geq}
     \Omega\lt(\frac{\delta_i r\eps m_{\chi} \|A_i-A_j\|}{d}\rt).
\end{align*}
Denote by $\chi_i=\chi(\langle \ell^*,A_i\rangle)$ the expected reward for action $A_i$. It remains to upper-bound
\[
    \bbE\lt[
    1_{A^*=A_i}
    \cdot
    \lt|
    \bbE^{t}
    [\chi_i-\chi_j]
    -(\chi_i-\chi_j)
    \rt|
    \rt].
\]
Note that the condition of Theorem~\ref{thm:GLM-MLE-subgaussian} holds for $\delta=e^{-d^2}$.
By Theorem~\ref{thm:GLM-MLE-subgaussian} and Lemma~\ref{lem:freqguarantee}, the value $\frac{m_{\chi}\gamma^{1/2}\big(\bbE^{t}[\chi_i-\chi_j]-(\chi_i-\chi_j)\big)}{\|A_i-A_j\|}$ is centered and uniformly bounded, up to modification on an event with probability $e^{-\Omega(d^2)}$.
Using Lemma~\ref{lem:subgaussian-tail} and the fact that $\chi$ is $M_{\chi}$-Lipschitz, we find
\begin{align*}
    \bbE\lt[
    1_{A^*=A_i}
    \cdot
    \lt|
    \bbE^{t}
    [\chi_i-\chi_j]
    -(\chi_i-\chi_j)
    \rt|
    \rt]
    \leq
    O\lt(
    \frac{\delta_i M_{\chi}  \|A_i-A_j\|}{m_{\chi}}
    \sqrt{\frac{\log(1/\delta_i)}{\gamma}}+e^{-d^2}\rt).
\end{align*}
Since we assumed $\gamma\geq \frac{C M_{\chi}^2 d^3 \log(4/r\eps)}{r^2 m_{\chi}^4 \eps^2}$
for a large enough absolute constant $C$, we conclude via:
\begin{align*}
    \bbE\lt[
    1_{A^*=A_i}
    \cdot
    \lt|
    \bbE^{t}
    [\chi_i-\chi_j]
    -(\chi_i-\chi_j)
    \rt|
    \rt]
    &\leq
    O\lt(\frac{\delta_i M_{\chi}
    \|A_i-A_j\|}{m_{\chi}}
    \sqrt{\frac{\log(1/\delta_i)}{\gamma}}+e^{-d^2}\rt)
    \\
    &\leq
    \Omega\lt(\frac{\delta_i r\eps m_{\chi} \|A_i-A_j\|}{d}+e^{-d^2}\rt)
    \\
    &\leq
    \bbE\Big[1_{A^*=A_i}\cdot \big[\chi(\langle \ell^*,A_i\rangle)-\chi(\langle \ell^*,A_j\rangle)\big]\Big]
\end{align*}
In the second step we used $\log(1/\delta_i)\leq d\log(4/r\eps)$ which follows from \eqref{eq:deltai-LB}.
At the end we used the assumption that $r,\eps,m_{\chi}\geq e^{-d}$ to absorb the $e^{-d^2}$ additive term.
\end{proof}

\section{Proofs for Counterexamples}

\begin{proof}[Proof of Proposition~\ref{prop:counterexample}]
    We show that \eqref{eq:counterexample} holds conditionally on $A^{(1)}\neq A_3$, and that equality holds conditionally on $A^{(1)}=A_3$. Together these imply the result. 
    
    First if $A^{(1)}=A_1$ or $A^{(2)}=A_2$, then at time $2$ we have learned the first coordinate of $\ell^*$. So in this case, if $\bbP^2[A^*=A_2]$ then $\bbE^2[\ell^*=(\mu_1,0)]$ for some $\mu_1>0$ almost surely. In particular, on this event 
    \[
        \bbE^2[\langle \ell^*,A_3-A_1\rangle]>0
    \]
    so we conclude that 
    \begin{align*}
    \bbE[\mu_3~|~(A^{(1)},A^{(2)})=(A_1,A_1)]&>\bbE[\mu_1~|~(A^{(1)},A^{(2)})=(A_1,A_1)],
    \\
    \bbE[\mu_3~|~(A^{(1)},A^{(2)})=(A_2,A_1)]&>\bbE[\mu_1~|~(A^{(1)},A^{(2)})=(A_2,A_1)].
    \end{align*}
    Next if $A^{(1)}=A_3$, then at time $2$ we have learned the value $\mu_3=\langle \ell^*,A_3\rangle$. We claim for any $\mu_3$,
    \begin{equation}
    \label{eq:isoceles}
    \bbP^2[A^*=A_1~|~\mu_3]=\bbP^2[A^*=A_1~|~-\mu_3].
    \end{equation}
    Indeed it is easy to see that $A^*=A_1$ if and only if $\ell^*$ lies in the angles spanned by the outward normal cone to $A_1$ for the triangle $A_1A_2A_3$. In other words, we must have
    \[
    \arg(\ell^*)\in \big[-\pi/2,\tan^{-1}(-4/3)\big].
    \]
    By construction, $\|A_1\|=\|A_1-A_3\|$ and so the angle bisector to this normal cone is orthogonal to $A_3$ by elementary geometry. It is easy to see that \eqref{eq:isoceles} now follows, and implies 
    \begin{align*}
    \bbE[\mu_3~|~(A^{(1)},A^{(2)})=(A_3,A_1)]
    =
    \bbE[\mu_1~|~(A^{(1)},A^{(2)})=(A_3,A_1)].
    \end{align*}
    Averaging over the conditioning on $A^{(1)}$ now yields the result (since $\bbP[A^{(1)}=A_1]>0$, say).
\end{proof}

\subsection{Exponential Lower Bound for Initial Exploration}

First, we observe that setting $\cA=\cE(\cP)$ to be the set of extreme points of $\cP$ is essentially equivalent to taking $\cA=\cP$ itself. The proof is easily seen to generalize to any convex polytope. This means that our lower bound below applies to the classes of convex action sets as well as separated action sets.

\begin{proposition}
\label{prop:extreme-points-reduction}
    Suppose there exists a BIC algorithm $\bbA$ which explores in $\cP$ and achieves $\gamma$-spectral exploration almost surely within $T$ timesteps. Then there exists a BIC algorithm $\bbA'$ which explores in $\cA=\cE(\cP)$ and achieves $\gamma$-spectral exploration almost surely within $Td$ timesteps.
\end{proposition}

\begin{proof}
    We show how to simulate each step of $\bbA$ using exactly $d$ steps of $\bbA'$, in a BIC way. For each $t\in \{0,1,\dots,T-1\}$, at the start of time-step $td+1$ we will have a simulated version of $\bbA$ which has received some artifical feedback and recommends an action $A^{(t)}\in \cP$. To construct $\bbA'$ we do the following:
    \begin{enumerate}
    \item Write $A^{(t)}=\sum_{i=1}^d p_{t,i}\wt A_{t,i}$ as a convex combination of $d$ not-necessarily-distinct $\wt A_{t,i}\in\cA$, with $p_{t,i}\geq 0$ and $\sum_{i=1}^d p_{t,i}=1$. 
    \item At timestep $td+i$ for $i\in \{1,2,\dots,d\}$, play action $\wt A_{t,i}$ and receive reward $\wt r_{t,i}\in \{0,1\}$.
    \item After timestep $(t+1)d$, choose $r_t=\wt r_{t,i}$ with probability $p_{t,i}$ independently of the past, and take this as the feedback for $\cA$.
    \end{enumerate}
    It is easy to see that $r_t$ has the correct expected value by linearity of rewards for the linear bandit problem. It then follows that after receiving $(r_1,\dots,r_{t-1})$, the recommendation $A^{(t)}\in\cP$ is BIC within $\cP$ by the BIC property of $\bbA$. By linearity it follows that conditioned on observing $A^{(t)}$, each extreme point $\wt A_{t,i}$ is also a BIC recommendation, since they all must have the same conditional mean reward. 
    
    It remains to show that $\bbA'$ also achieves $\gamma$-spectral exploration. Indeed
    \[
    \sum_{i=1}^d p_{t,i}\wt A_{t,i}^{\otimes 2}\succeq (A^{(t)})^{\otimes 2}
    \]
    follows by testing against any $v^{\otimes 2}$ and using Cauchy-Schwarz. This completes the proof.
\end{proof}

Based on Proposition~\ref{prop:extreme-points-reduction}, we state Proposition~\ref{prop:linear-bandit-LB} below in the discrete action set setting $\cA=\cE(\cP)$ and an assumption of $0.01$-BIC, with the understanding that it extends also to the convex action set $\cA=\cP$ for BIC algorithms. Note that for $\cP$ as above, all extreme points satisfy $x_d=0$ except for $\pm\hat A = \lt(0,0,\dots,0,\pm \frac{1}{10}\rt)$, and so $\gamma$-spectral exploration for any $\gamma>0$ requires exploring these actions. Our prior distribution $\ell^*\sim\mu$ will be uniform over the ``biased'' convex body $\cK/\sqrt{d}$ for
\[
    \cK
    \equiv
    [-0.5,1]^{d-1} \times [-1,1].
\]

\begin{proof}[Proof of Proposition~\ref{prop:linear-bandit-LB}]
    We set $\eps=0.01$.
    Consider a modification $\bbA'$ of $\bbA$ which changes all plays of $\pm\hat A$ to the prior-optimal action $A_1=\frac{(1,1,\dots,1,0)}{2\sqrt d}$.
    If $\bbA'$ and $\bbA$ play different actions at $m$ times in expectation, then by definition the expected total rewards satisfy
    \begin{equation}
    \label{eq:eps-BIC-aggregate}
    R_T(\bbA)+\eps m\geq R_T(\bbA').
    \end{equation}
    Recall that one of $\pm\hat A$ must be explored at least once regardless of $\ell^*$. Moreover at the first such modification time (denoted by $\tau$),
    \[
    \bbE^{\tau}[\langle \ell^*,\hat A\rangle]
    =
    \bbE^{\tau}[\langle \ell^*,-\hat A\rangle]
    =0
    .
    \]
    This implies 
    \begin{align*}
    \bbE[r_{\tau}(\bbA)+\eps - r_{\tau}(\bbA')]
    \leq 
    \eps-\bbE[\langle \ell^*,A_1\rangle]
    \leq \eps-\frac{d-1}{8d}
    \leq -1/9.
    \end{align*}
    Moreover $\bbA'$ makes at most $T-1$ additional modifications; at each time $s$ that such a modification occurs, 
    \[
    \bbE^s[r_{s}(\bbA)+\eps - r_{s}(\bbA')]
    \leq 
    \Big(
    |\bbE^s[\langle \ell^*,\hat A\rangle]|
    +\eps
    -
    \bbE^s[\langle \ell^*,A_1\rangle]
    \Big)_+.
    \]
    Summing over times $s$ and applying Jensen in the second inequality, we find
    \begin{align*}
    R_T(\bbA)+\eps m- R_T(\bbA')
    &\leq
    -1/9 
    +
    \sum_{s=1}^T
     \bbE\lt[
    \Big(
    |\bbE^s[\langle \ell^*,\hat A\rangle]|
    +\eps
    -
    \bbE^s[\langle \ell^*,A_1\rangle]
    \Big)_+
    ~\rt]
    \\
    &\leq
    -1/9
    +
    T\cdot 
    \bbE\lt[
    \Big(
    |\langle \ell^*,\hat A\rangle|
    +\eps
    -
    \langle \ell^*,A_1\rangle
    \Big)_+
    ~\rt]
    .
    \end{align*}
    A simple Chernoff estimate shows that 
    \begin{equation}
    \label{eq:exp-LB}
    \begin{aligned}
    \bbE\lt[
    \Big(
    |\langle \ell^*,\hat A\rangle|
    +\eps
    -
    \langle \ell^*,A_1\rangle
    \Big)_+
    ~\rt]
    \leq
    \bbE\lt[
    \Big(
    \frac{1}{10}
    -
    \langle \ell^*,A_1\rangle
    \Big)_+
    ~\rt]
    \leq 
    e^{-\Omega(d)}.
    \end{aligned}
    \end{equation}
    Recalling \eqref{eq:eps-BIC-aggregate}, we conclude $T\geq e^{\Omega(d)}$ as desired.
\end{proof}
\balance

\section{$j$-Recommendation Game}

We first formally define the $j$-recommendation game. We proceed more generally than in the main body, defining it relative to any $d$-tuple $(N_1,\dots,N_d)$. We recall the definition of the static $\sigma$-algebra $\mathcal G_{N_1,\dots,N_d}$ which is generated by $N_i$ samples of each atom $i$. When considering arm $j$, we will always have $N_k=0$ for all $k>j$. If $N_i=N$ for all $i$ we recover the $(j,N)$-recommendation game as defined in the main body. The definition of $(j,N)$-informed generalizes readily to $(j,\mathcal G)$-informed.

\begin{definition}
The \emph{$j$-recommendation game} is a two-player zero sum game played between a \emph{planner} and an \emph{agent}. The players share a common independent prior over the true mean rewards $(\theta_i)_{i\leq j}$, and the planner gains access to the static $\sigma$-algebra $\mathcal G=\mathcal G_{(N_1,\dots,N_j)}$. The planner then either:
\begin{enumerate}
    \item Recommends an arm $A_j\in \cA_j$ containing $A_j$.
    \item Does nothing.
\end{enumerate}
In the first case, the agent observes $A_j\in \cA_j$ and chooses a response $A_{-j}$ from a mixed strategy on $\cA_{-j}$ which can depend on $A_j$. The payoff for the planner is zero if the planner did nothing, and otherwise $\theta_{A_j}-\theta_{A_{-j}}$.
\end{definition}

\begin{definition}
    A planner strategy $\pi$ for the $j$-recommendation game is said to be $(j,\lambda)$-padded, for a function $\lambda:\cA_j\to\bbR_{\geq 0}$, if for each $A_j\in \cA_j$ and $A_{-j}\in\cA_{-j}$:
    \begin{equation}
    \label{eq:gain-per-Aj}
    \min_{A_{-j}\in \cA_{-j}}\bbE[(\theta_{A_j}-\theta_{A_{-j}})\cdot 1_{\pi=A_j}]= \lambda_{A_j}.
    \end{equation}
    Such a strategy has total $j$-padding at least 
    \[
    \lambda_j\equiv\sum_{A_j\in\cA_j}\lambda_{A_j}.
    \]
\end{definition}

Note that since $\lambda_{A_j}\geq 0$, a planner strategy need not be $(j,\lambda)$-padded for any $\lambda$, and hence need not have a total padding value.

We can without loss of generality view all $j$-recommendation game strategies as depending only on the posterior means of each arm conditional on $\mathcal G$. In the below, we set $\widetilde\theta_i=\mathbb E[\theta_i|\mathcal G]$ for the relevant static $\sigma$-algebra $\mathcal G$.  Given a planner strategy in the $j$-recommendation game, we naturally obtain a corresponding $(j,\mathcal G)$-informed policy for our original problem in which we recommend arm $A_j$ when the planner would, and recommend the $\mathcal G$-conditional-expectation-maximizing arm otherwise. The key point of this game is as follows.

\begin{lemma}
\label{lem:gameBIC}
If a strategy $\pi$ for the planner in the $j$-recommendation game has minimax value $\lambda_j$, then its total $j$-padding value equals $\lambda_j$.
\end{lemma}

\begin{proof}
Simply note that the left-hand side of \eqref{eq:gain-per-Aj} is the contribution of playing $A_j$ to the minimax value of $\pi$ for the planner.
\end{proof}

\begin{lemma}
\label{lem:bic-reduction}
Suppose there exists a planner strategy $\pi$ in the $j$-recommendation game with total $j$-padding at least $\lambda_j$. Let \msedit{$p\geq d/(d+\lambda_j)$} and let $b_j\sim \Ber(p)$ be independent of the signal. 

Consider a modified game where the planner observes $\cG$ if and only if $b_j=1$ (and where the agent never observes $b_j$). There exists a BIC planner strategy $\wt\pi$ in this game such that if $b_j=0$, then $\wt\pi$ always recommends an action in $\cA_j$.
\end{lemma}

\msedit{
\begin{proof}
The planner follows $\pi$ if $b_j=1$. Conditioned on $b_j=0$, the planner plays from the probability distribution $q$ on $\cA_j$ given by $q(A_j)=\frac{\lambda_{A_j}}{\lambda_j}$.
We call this modified strategy $\hat\pi$.
Noting that $\bbE[\theta_{A_j}-\theta_{A_{-j}}]>-d$, we find that if the agent observes $\hat\pi=A_j$, then $A_j$ has higher conditional mean reward than any $A_{-j}\in\cA_{-j}$:
\begin{align*}
    \bbE[(\theta_{A_j}-\theta_{A_{-j}})\cdot 1_{\hat\pi=A_j}]
    &=
    \bbE[(\theta_{A_j}-\theta_{A_{-j}})\cdot 1_{b_j=1,\pi=A_j}]
    +
    (1-p)q(A_j)\cdot 
    \bbE[\theta_{A_j}-\theta_{A_{-j}}]
    \\
    &>
    p\lambda_{A_j}
    -j
    d(1-p)q(A_j)
    .
\end{align*}
Since $p\geq d/(d+\lambda_j)$, it easily follows that the last expression is non-negative.

Therefore the greedy strategy conditioned on observing $\hat\pi=A_j$ is to choose some $A_j'\in \cA_j$ (any $A_{-j}\in \cA_{-j}$ is inferior to $A_j$ hence suboptimal). The claimed BIC strategy $\wt\pi$ is given by playing $A_j'$ when $\hat\pi=A_j$. Note that by definition, $\wt\pi$ exploits conditioned on the signal from $\hat\pi$, hence is BIC.
\end{proof}
}

Next we upper bound the number of samples required in $\mathcal G$ to ensure $\padding\geq\Omega(\padG)$. The point is that the game has an $N\to\infty$ limit that can be approximated by coupling. To this end, we let $\lambda_{j,\infty}$ be the minimax value for the $N=\infty$ version where the planner observes the exact values of $\theta_1,\dots,\theta_j$. We note that in the simpler setting of \cite{sellke-slivkins}, $\lambda_{j,\infty}=\inf_{q\in\Delta_j} \bbE[(\theta_j-\theta_q)_+]$ has a simple interpretation by using the minimax theorem to choose the agent's strategy first. It is similarly possible to apply the minimax theorem here but the result is not particularly interpretable since the planner chooses a different mixed strategy for each $A_j\in\cA_j$.

\begin{lemma}
\label{lem:chernoff2}
For $N\geq \Omega\lt(\frac{d^2 \log|\cA|}{\lambda_{j,\infty}^2}\rt)$, there is an $N$-informed policy with minimax value at least $\lambda_{j,\infty}/2$.
\end{lemma}

\begin{proof}
    Given the signal, first generate a posterior sample $(\hat\theta_1,\dots,\hat\theta_j)$ for $\theta_1,\dots,\theta_j$ and then consider the infinite-sample policy $\pi_{\infty}$. 
    We claim that $\pi_{\infty}$ has minimax value at least $\lambda_{j,\infty}/2$

    To find the minimax value of this policy, we use Lemmas~\ref{lem:freqguarantee} and \ref{lem:subgaussian-tail}. In particular they imply that for each $(A,A_j)$,
    \begin{align*}
    \bbE[|\hat\theta_A-\theta_A|\cdot 1_{\pi=A_j}]
    &\leq
    CdN^{-1/2}
    p_{\pi}(A_j)
    \sqrt{\log(1/p_{\pi}(A_j))}.
    \end{align*}
    Taking $A=A_j$ and $A=A_{-j}$, we find that the minimax value of this policy is at least
    \begin{align*}
    \lambda_{j,\infty} - O(dN^{-1/2})\cdot \sum_{A} p_{\pi}(A)\sqrt{\log(1/p_{\pi}(A))}
    \geq
    \lambda_{j,\infty}
    -
    O(dN^{-1/2}\sqrt{\log|\cA|}).
    \end{align*}
    \msedit{
    The last expression is at least $\lambda_{j,\infty}/2$} for $N\geq \Omega\lt(\frac{d^2 \log|\cA|}{\lambda_{j,\infty}^2}\rt)$ as desired.
\end{proof}

\begin{proof}[Proof of Lemma~\ref{lem:BIC-growth}]
    \msedit{
    By Lemma~\ref{lem:chernoff2}, conditioned on $z=1$ the minimax value of the game is at least $\lambda_{j,\infty}/2$. 
    We now apply Lemma~\ref{lem:bic-reduction}, with $b_j$ the signal $z$. 
    The policy $\wt\pi$ is the desired one, proving the lemma.
    (Note that by definition $\lambda=\ulambda/2d$ in Lemma~\ref{lem:BIC-growth}.)
    }
\end{proof}

\begin{proof}[Proof of Proposition~\ref{prop:semibandit-LB}]
    We consider two agents. The first, an \textbf{obedient agent}, simply obeys recommendations. The second, a \textbf{$j$-avoiding agent}, when recommended to play some $A_j\in \cA_j$, instead plays from $\cA_{-j}$ using the minimax optimal response for the easy $j$-recommendation game. We will show that if $T$ is too small and almost surely some $A_j\in\cA_j$ must be recommended, then the $j$-avoiding agent attains larger expected reward than the obedient one against any planner. This will contradict the BIC property. First note that each time-step, by definition the alternative agent loses at most $\olambda_{j}$ per round compared to an obedient agent. To get the lower bound, we will show that the alternative agent does better the first time $t_j$ that $A_j$ is recommended (i.e. the random value $t_j$ is minimal such that $A_j\in A_{t_j}$). 

    At this time, note that $\theta_{A_j}-\bbE[\theta_{A_{-j}}~|~A_j]$ has standard deviation at least $\sigma$ for each $A_j\in\cA_j$ and any agent strategy. ($A_j$ can depend on $\theta_1,\dots,\theta_{j-1}$ but is independent of $\theta_j$.) In particular, 
    \begin{align*}
    \bbE\lt[\lt(\theta_{A_{t_j}}-\bbE[\theta_{A_{-j}}~|~A_{t_j}]\rt)_+\rt]
    +
    \bbE\lt[\lt(\theta_{A_{t_j}}-\bbE[\theta_{A_{-j}}~|~A_{t_j}]\rt)_-\rt]
    &=
    \bbE\lt[\lt|\theta_{A_{t_j}}-\bbE[\theta_{A_{-j}}~|~A_{t_j}]\rt|\rt]
    \\
    &\geq
    \frac{1}{d}\bbE\lt[\lt(\theta_{A_{t_j}}-\bbE[\theta_{A_{-j}}~|~A_{t_j}]\rt)^2\rt]
    \\
    &\geq \sigma^2/d. 
    \end{align*}
    The first term is at most $\olambda_j$ for any strategy, so
    \begin{align*}
    \bbE\lt[\theta_{A_{t_j}}-\bbE[\theta_{A_{-j}}~|~A_{t_j}]\rt]
    &
    =\bbE\lt[\lt(\theta_{A_{t_j}}-\bbE[\theta_{A_{-j}}~|~A_{t_j}]\rt)_+\rt]
    -
    \bbE\lt[\lt(\theta_{A_{t_j}}-\bbE[\theta_{A_{-j}}~|~A_{t_j}]\rt)_-\rt]
    \\
    &
    \leq 2\olambda_j-\frac{\sigma^2}{d}.
    \end{align*}
    Thus in total, the $j$-avoiding agent outperforms the obedient one by
    \[
    \frac{\sigma^2}{d}-(T+2)\olambda_j,
    \]
    which is non-positive if recommendations are BIC.
    This implies the claimed estimate.
\end{proof}

\end{document}